\documentclass[11pt]{article}

\usepackage{times}
\usepackage{graphicx}
\usepackage{amssymb}
\usepackage{amsthm}
\usepackage{amsmath}
\usepackage[makeroom]{cancel}
\usepackage{diagbox}
\usepackage{url}
\usepackage{fullpage}

\begin{document}
\title{\Large A Formal Approach to Power Optimization in CPSs with Delay-Workload Dependence Awareness\thanks{
An abridged version of this paper is to appear in a Special Issue of IEEE Transactions on Computer-Aided Design of Integrated Circuits and Systems~\cite{Journalversion}; a preliminary version of this work by Yang and Ha~\cite{yang2015modeling} was presented at ISLPED 2015.
}}
\author{Hyung-Chan An\thanks{
{\tt hyung-chan.an@yonsei.ac.kr}. Department of Computer Science, Yonsei University. Part of this work was conducted while the author was with \'Ecole \mbox{Polytechnique} F\'ed\'erale de Lausanne. Supported in part by ERC Starting Grant 335288-OptApprox.
}
\and
Hoeseok Yang\thanks{
{\tt hyang@ajou.ac.kr}. Department of Electrical and Computer Engineering, Ajou University. Supported by ICT R\&D program of MSIP/IITP (B0101-15-0661, the research and development of the self-adaptive software framework for various IoT devices).
}
\and
Soonhoi Ha\thanks{
{\tt sha@snu.ac.kr}. Department of Computer Science and Engineering, Seoul National University. Supported by Basic Science Research Program (NRF-2013R1A2A2A01067907) through NRF Korea funded by MSIP.
}
}
\date{}

\maketitle

\newcommand{\wbase}{\ensuremath{w_{\mathrm{base}}}}
\newcommand{\POL}{\mathsf{POL}}
\newcommand{\OPT}{\mathsf{OPT}}
\newcommand{\spow}{P}
\newcommand{\eps}{\epsilon}
\newcommand{\mystar}{*}
\newcommand{\bigmid}{\mathrel{}\middle|\mathrel{}}
\newcommand{\argmin}{\operatornamewithlimits{argmin}}
\newcommand{\cl}{\mathrm{cl}}
\newcommand{\baseindex}{\eta}
\newcommand{\mypar}[1]{\paragraph*{#1}}
\newcommand{\api}[1]{\emph{#1}}
\newcommand{\vdd}{v_{\mathrm{dd}}}
\newcommand{\vth}{v_{\mathrm{th}}}

\newtheorem{thm}{Theorem}
\newtheorem{lemma}{Lemma}
\newtheorem{cor}{Corollary}
\newtheorem{defn}{Definition}
\newtheorem{obs}{Observation}
\newtheorem{claim}{Claim}
\newtheorem{conj}{Conjecture}

\setcounter{page}{0}
\maketitle
\thispagestyle{empty}

\begin{abstract}
The design of cyber-physical systems (CPSs) faces various new challenges that are unheard of in the design of classical real-time systems. Power optimization is one of the major design goals that is witnessing such new challenges. The presence of interaction between the cyber and physical components of a CPS leads to dependence between the time delay of a computational task and the amount of workload in the next iteration. We demonstrate that it is essential to take this delay-workload dependence into consideration in order to achieve low power consumption.

In this paper, we identify this new challenge, and present the first formal and comprehensive model to enable rigorous investigations on this topic. We propose a simple power management policy, and show that this policy achieves a best possible notion of optimality. In fact, we show that the optimal power consumption is attained in a ``steady-state'' operation and a simple policy of finding and entering this steady state suffices, which can be quite surprising considering the added complexity of this problem. Finally, we validated the efficiency of our policy with experiments.
\end{abstract}


\newpage

\section{Introduction}
\label{sec:intro}

Modern computation is not confined to small silicon dice anymore. In \emph{cyber-physical systems} (CPSs), computers interact with the physical world: actuators allow computer systems to manipulate the physical world, while sensors constantly provide the computer systems with outside information \cite{lee2008cyber,wu2011wireless}. The resulting feedback control loop makes the hallmark of CPSs, and poses unique challenges in the design of CPSs that were unheard of in the design of \emph{classical} real-time systems. Among these, this paper focuses in particular on the new challenges faced in power optimization in CPSs.

The existence of interaction between the \emph{cyber} and \emph{physical} components of CPSs implies that the physical world affects not only the particular value of the input to the computational task but also the amount of the computational workload. In particular, the time delay of a computational task can change the amount of workload in the next iteration: a typical example is witnessed when a CPS maintains an internal model of the physical world, where the model can range from a simple snapshot of a sensor reading \cite{wu2011wireless} to a sophisticated model of beliefs about the external world \cite{oh2007tracking,kubota2006cooperative}. The longer an iteration goes, the further this model can drift away from the physical reality; regaining its accuracy in the next iteration incurs extra computational cost.

One of the tools that are widely used by CPSs is computer vision-based object tracking algorithms. In a CPS equipped with object tracking capability, its \emph{internal model} of the physical world will maintain the coordinate of the tracked object in the image. As the execution delay between two consecutive invocations of the tracking algorithm becomes longer, the algorithm will be required to search a larger area in order to reconcile the physical reality with the internal model. It is important to properly address the resulting delay-workload dependence, since object tracking is frequently used in a variety of CPSs, including, but not limited to, vision-assisted control of unmanned air vehicles (UAV)~\cite{1656418, 5509287, kanade2004real}, surveillance camera tracking \cite{coifman1998real, 6155714}, and augmented reality \cite{sonntag2013towards, you2001fusion}. See \cite{ganapathi2010real, Ganapathi:2012:RHP:2403205.2403261} for additional examples of vision-based algorithms whose workload may vary with the execution delay.

Delay-workload dependence manifests itself also in different types of CPSs. Agrawal et al. \cite{Agrawal:2008:EPM:1376616.1376634}, for example, study the optimization of pattern matching over event streams, where the queries can be handled either by small amount (shorter delay, less workload) or in an aggregated manner (longer delay, more workload). CPSs whose cyber component exploits temporal coherence bears delay-workload dependence: any iterative algorithms that can be warm-started can lead to one. Haptic rendering in Human-Computer Interface (HCI) is an example, as it often uses adaptive sampling techniques to deal with the stringent real-time constraint~\cite{Barbic2007} and the rendering algorithm can be warm-started to exploit the temporal coherence~\cite{Barbic08}.

In order to ensure the responsiveness of a real-time system (and the physical stability of a CPS), real-time constraints are specified, often in the form of a maximum latency. Once this constraint is given, optimizing the microprocessor's power consumption over the operation frequency is rather straightforward in a classical real-time system where delay-workload dependence is absent: one can choose the running frequency of the microprocessor as low as possible without violating the latency constraint, and set the voltage of the microprocessor to the minimum needed to run at the chosen frequency \cite{aydin2004power}.

However, the presence of the delay-workload dependence invalidates this straightforward strategy. Running at an excessively low speed in one iteration would lead to an unfavorable increase in the workload of the next iteration, thereby requiring the next iteration to run at a high speed; running at an excessively high speed on the other hand could simply result in suboptimal power consumption. This dilemma creates a new need for a ``smart'' power optimization strategy that is aware of the delay-workload dependence.

\newpage
In this paper, we\begin{itemize}
\item identify this newly posed challenge, and present the first formal and comprehensive model which enables a rigorous analysis;
\item propose a \emph{simple} power management policy;
\item show the optimality of our policy;
\item and experimentally evaluate its efficiency.
\end{itemize}

In developing our results, we aim at formulating a model that is as general as possible: in fact, we do not assume any particular data representation on the parameters of our model. Our power management policy is therefore given as a mathematical characterization rather than an algorithmic procedure. In spite of the significant added complexity to the problem, our power management policy remains surprisingly simple; this simplicity enables in many cases an algorithmic reinterpretation of the characterization of our policy. However, we will adhere to its mathematical characterization throughout this paper, in favor of generality. 
This is particularly natural considering that the proposed technique is primarily a design-time methodology.

In Section~\ref{sec:model}, we present our formal model and formulate the problem of power optimization as a concrete mathematical question. Then, in Section~\ref{sec:opt}, we present our power management policy under a simplifying assumption that the workload is given as a continuous function of execution delay. Whilst this assumption is \emph{not} mathematically necessary, it will simplify the analysis of our policy and make the underlying intuition more visible.
The proposed model can be further generalized to be applied to a wide spectrum of potential applications; this versatility is discussed in Section~\ref{sec:modelgen}.  The analysis in full generality is presented in Section~\ref{sec:genproof}. Section~\ref{sec:case} evaluates the practical efficiency of our power management policy. To this end, we experimentally measured the power consumption characteristics of a mobile computing platform Samsung Exynos5422 (in Section~\ref{subsec:env}), profiled an OpenCV-based object tracking application to quantitatively identify the delay-workload dependence (Section~\ref{subsec:w_exp}), and evaluated the efficiency of our power management policy (Section~\ref{subsec:p_exp}).

\subsection{Related Work}

There have been some studies to apply multiple operation \emph{modes} in CPS design. Jha et al. \cite{jha2010synthesizing} studied a system that has different execution modes, each of which is known and modeled as a node in the modeling automaton. Canedo et al. presented a context-sensitive synthesis of CPS \cite{canedo2013context}. In order to overcome the incompleteness of the functional model, they adopt the recycle function that reliably generates the simulation model based on the context that the previous results caused. A runtime optimization of CPS is proposed by Cao et al. \cite{cao2013online},
where the design parameters are adaptively tuned considering the feedback results. While CPSs are enriched or optimized via multiple operation modes in the above mentioned works, none of them has taken the execution delay into consideration as a source of variance in workload.

There are a handful of literatures that study the relationship between control stablity and system performance in control-centric CPSs. A design guideline for flexible delay constraints in distributed CPS was proposed by Goswami et al. \cite{goswami2011co, goswami2014relaxing}, where some of the samples are allowed to violate the given delay deadline. They presented the applicability of the proposed approach using the FlexRay dynamic segment as a communication medium. They could improve the resource efficiency or flexibility of CPSs in favor of the stability. However, this relaxation of design constraints is not always feasible. Zhang et al. \cite{zhang2008task} took advantage of the fact that the longer computation delay may lower the gain of the control algorithm in the control example of inverted pendulums. If the delay becomes longer, the system can support more inverted pendulums within a given resource at the cost of reduced control stability. Such co-design approaches of control algorithm and system, though, are still blind to the relation between execution delay and workload in CPSs, which differentiates the proposed method from them. 

An alternative approach to the co-design of algorithm and system is found in the application of anytime algorithm to control-centric systems \cite{bhattacharya2004anytime, fontanelli2008anytime, quevedo2013sequence}. Anytime algorithms are a kind of algorithms that can be completed arbitrarily at any point and the quality of the algorithm output is proportional to the spent time. That is, the amount of time invested in the cyber system, in this model, is directly coupled with the stability of the system. However, how this compromised stability affects the system in ever-present feedback loops in CPSs is still missing in their models. On the contrary, in the proposed model, the harmed stability due to the lengthened delay manifest itself in the increased workload at the successive iteration.

\section{Problem Formulation}
\label{sec:model}

In this section, we present our formal model of the power optimization problem in the presence of delay-workload dependence.
Firstly, we describe our model and introduce the notation to be used throughout this paper in Section~\ref{subsec:model}.

Section~\ref{sec:aopt} then introduces a best-possible notion of optimality, called \emph{asymptotic optimality}.
Finally, the full formulation of the power optimization problem is formally presented in Section~\ref{subsec:prob}.


\subsection{Model}
\label{subsec:model}

\mypar{Units of measurement}
In order to keep the presentation as succinct as possible, we will choose the units of measurement in a careful way. Note that these choices are purely for the sake of notational convenience and do not inherently change the analysis: our entire result can be presented under any arbitrary choice of units by introducing appropriate conversion factors.

First, we choose the unit of processing speed so that the maximum speed corresponds to one unit. For instance, if the given CPS is equipped with a microprocessor with the maximum operating frequency of 2GHz, running it at 1GHz is denoted by $s=0.5$.

We also need to choose a unit of workload; we define one unit of workload as the workload that can be processed in one unit of time at the full speed ($s=1$). For example, if the system runs at the speed of $1/2$, it will take two units of time to process one unit of workload.

The notation to be defined in this section and Section~\ref{sec:opt} is summarized in Table~\ref{tab:not}.

\begin{table}[t]
\centering
\caption{Summary of the notation in Sections~\ref{sec:model} and~\ref{sec:opt}}
\label{tab:not}
\label{my-label}
\vspace{1ex}\small
\begin{tabular}{c|l}
\hline
\multicolumn{1}{c|}{Notation} & \multicolumn{1}{c}{Meaning} \\ \hline\hline
$\spow$ & power consumption characteristics\\ \hline
$W$ & delay-workload relation \\ \hline
$w_i$ & workload of the $i$-th iteration  \\ \hline
$w_1$ & initial workload (i.e., workload of the first iteration) \\ \hline
$\wbase$ & baseline workload (bookkeeping work that is always required) \\ \hline
$s_i$ &  processing speed of the $i$-th iteration \\ \hline
$t_i$ & execution delay of the $i$-th iteration  \\ \hline
$T$ & real-time constraint (i.e., maximum allowed execution delay)  \\ \hline
$n$ & time horizon  \\ \hline \hline
$t_{\min}$ & minimum possible execution delay \\ \hline
$\hat s$ & target speed \\ \hline
$\tau$ & ideal point workload \\ \hline
\end{tabular}
\end{table}

\subsubsection{Power consumption characteristics}
Modern microprocessors support dynamic voltage and frequency scaling (DVFS), where the operating voltage and frequency can be modulated to optimize the power consumption. In describing the power consumption characteristics, our model does not assume any specific DVFS model; instead, it achieves higher generality by describing the characteristics with a function that satisfies a small set of natural axioms. This allows us to use our model when the power consumption is adjusted via mechanisms other than DVFS, such as processor selection in a heterogeneous multi-processor system. For $\spow:[0,1]\to\mathbb{R}_+$, let $\spow(s)$ denote the power consumption when the system is run at speed $s$. Note that $\spow$ is a function of speed only: if the power consumption is adjusted via DVFS, the operating voltage can be optimally chosen once the frequency is fixed, and this choice can be implicitly encoded within the definition of $\spow$.

The axiomatic assumptions we will make on $\spow$ is as follows. First, we assume that $\spow$ is nondecreasing, convex, and continuous\footnote{This simply amounts to assuming $\lim_{f\downarrow 0}\spow(f)=\spow(0)$ and $\lim_{f\uparrow 1}\spow(f)=\spow(1)$, since $\spow$ is already convex.}. Note that this set of assumptions is general enough to embrace, for example, the power consumption characteristics of CMOS circuits: the power dissipation of CMOS gates is dominated by $c\vdd^2f$, where $c$ denotes the load capacitance, $\vdd$ is the operating voltage, and $f$ is the frequency; the maximum operating frequency is given proportional to $(\vdd-\vth)^2/\vdd$, where $\vth$ is the threshold voltage. From these facts, it is easy to see that the resulting $\spow$ is nondecreasing, convex, and continuous. However, instead of relying on such an idealized formula, our model can be used with actual power consumption characteristics obtained by measurement or taken from manufacturer's data. We also remark that these assumptions are far from being minimal: in Section~\ref{sec:modelgen}, we show how some of these assumptions can be dropped without loss of generality.

Our second assumption is that the operating frequency can be modulated to any given value in $[0,1]$. This assumption, however, can also be removed to handle the case where the microprocessor has only a few predetermined modes of operation. Details can be found in Section~\ref{sec:modelgen}.

\subsubsection{Delay-workload dependence}
We model the delay-workload relation as the following function. Let $T$ be the maximum execution delay as set out by the real-time constraints; the delay-workload relation is specified by $W:(0,T]\to\mathbb{R}_+$, where $W(t)$ denotes the workload of the iteration that follows an iteration of execution delay $t$. The workload of the first iteration is denoted by $w_1$. In practice, one can determine $w_1$ and $W$ for the application at hand by using static analysis or profiling techniques at design time.

As was discussed in Section~\ref{sec:intro}, the main difficulty of the power optimization problem lies in the fact that a longer delay leads to a larger workload in the next iteration. We will thus assume that $W$ is a nondecreasing function.\footnote{This assumption, though, can be replaced. Details follow in Section~\ref{sec:modelgen}.} In every iteration, there would be some basic bookkeeping work required regardless of the previous iteration's delay. Let $\wbase>0$ denote the workload arising from such basic work, and we will have $W(t)\geq \wbase$ for all $t\in(0,T]$, and $w_1\geq\wbase$.

\subsubsection{Execution trace}
Suppose that the CPS runs for $n$ iterations. We call $n$ the \emph{time horizon} of the system, but we will not assume any \emph{a priori} knowledge about the time horizon: the system does not know in advance when it will be halted from outside. Let $w_i$, $s_i$, and $t_i\ (1\leq i\leq n)$ be the workload, processing speed, and execution delay of the $i$-th iteration, respectively.

We do not allow changing the processing speed within a single iteration, and this assumption does not harm the power optimality due to the convexity of $\spow$.
If the processing speed changes within an iteration, we can instead fix the speed to the average speed during that iteration and we will be able to process the same amount of work while consuming no more power\footnote{Suppose that the processing speed changes during the $i$-th iteration. For \mbox{$s_i:[0,t_i]\to[0,1]$}, let $s_i(t)$ denote the processing speed after $t$ units of time since the beginning of the $i$-th iteration. Let $\bar s_i:=\frac{1}{t_i}\int_0^{t_i}s_i(t)dt$. We have that the average power consumption $\frac{1}{t_i}\int_0^{t_i}\spow(s_i(t))dt$ is greater than or equal to $\spow(\bar s_i)$ from Jensen's inequality. (We assume that the integrals exist.)}. 

For all $i$, $w_i$ by definition has to be less than or equal to $s_it_i$: as the $i$-th iteration runs for $t_i$ units of time at speed $s_i$, at most $s_it_i$ units of workload can be processed by the end of this iteration, whereas $w_i$ is defined as the amount of work that needs to be done in the $i$-th iteration. 
In fact, we can further assume that they are equal, i.e., $w_i=s_it_i$, in a power-optimal scenario. Suppose that $w_i<s_it_i$ for some iteration $i$. We can then decrease $s_i$ to $\frac{w_i}{t_i}$ instead; this will not increase the power consumption of iteration $i$ due to the monotonicity of $\spow$ and will not otherwise change the system's behavior.

An \emph{execution trace} is defined as a sequence of processing speeds.
\begin{defn}[Execution trace]
\label{defn:et}
Suppose that a given CPS has run for $n$ iterations. We call the sequence $s_1,\ldots,s_n$ of the processing speeds its \emph{execution trace}, and $n$ the \emph{length} or \emph{time horizon} of this execution trace.
\end{defn}
An execution trace contains sufficient information to determine the execution delay and workload of every iteration given the parameters of the CPS: we have $t_i:=\frac{w_i}{s_i}$ and $w_{i+1}=W(t_i)$ for all $i$. Note that the real-time constraints demand that $t_i\leq T$ for all $i$. We will sometimes call an execution trace a \emph{real-time feasible execution trace} in order to emphasize the presence of the real-time constraints.

\subsection{Asymptotic Optimality}
\label{sec:aopt}

In this subsection, we show that a ``natural'' notion of optimal power management policy is an impossible goal to achieve, and introduce \emph{asymptotic optimality} as the ``right'' notion of optimality.

Given a system specified by its power consumption characteristics $P$, delay-workload relation $W$, initial workload $w_1$, and the real-time constraint $T$, we could set our goal as designing a power management policy such that, if the system is halted after $n$ iterations, the resulting execution trace minimizes the average power consumption\[
\frac{\sum_{i=1}^n t_i\spow(s_i)}{\sum_{i=1}^n t_i}
.\]Unfortunately, however, this goal is impossible to achieve.

Our model does not assume that we ``know the future'', so the power management policy needs to work without knowing when the system is to be halted. This makes it impossible for a power management policy to produce an exactly optimal execution trace, which is demonstrated by the following toy example: consider a system with $T=1$, $w_1=\frac{1}{2}$, $W(t)=\sqrt{t}$, and $\spow(s)=s^2$. The (unique) execution trace of length $2$ that minimizes the average power consumption in this system is $s_1\approx 0.6180$, $s_2\approx 0.8995$.\footnote{
Note that the average power consumption is given as $\frac{t_1P(s_1)+t_2P(s_2)}{t_1+t_2}=\frac{\frac{\frac{1}{2}}{s_1}s_1^2+\frac{\sqrt{\frac{\frac{1}{2}}{s_1}}}{s_2}s_2^2} {\frac{\frac{1}{2}}{s_1}+\frac{\sqrt{\frac{\frac{1}{2}}{s_1}}}{s_2}}=:p(s_1,s_2)$ and that $s_1$ and $s_2$ are subject to the following constraints: $0\leq s_1,s_2\leq 1$, $\frac{\frac{1}{2}}{s_1}\leq 1$, and $\frac{\sqrt{\frac{\frac{1}{2}}{s_1}}}{s_2}\leq 1$. The last two constraints are the real-time constraints. Since $p(s_1,s_2)=\frac{\frac{s_1}{2}+\sqrt{\frac{1}{2s_1}}s_2} {\frac{1}{2s_1}+\frac{1}{s_2}\sqrt{\frac{1}{2s_1}}}$ is a nondecreasing function of $s_2$, its minimum is attained when $s_2=\sqrt{\frac{1}{2s_1}}$, making the last constraint tight. In this case, we have $p(s_1,s_2)=\frac{\frac{s_1}{2}+\frac{1}{2s_1}}{\frac{1}{2s_1}+1}$, which in turn is minimized by $s_1=\frac{\sqrt{5}-1}{2}\approx 0.6180$. Note that $\sqrt{\frac{1}{2s_1}}=\sqrt{\frac{\sqrt{5}+1}{4}}\approx 0.8995$.
} Since the power management policy does not know the time horizon in advance, it would need to choose $s_1$ and $s_2$ as the speed in the first two iterations in order to successfully produce an optimal execution trace in case the system is halted after two iterations. However, if the system is halted after three iterations, the resulting execution trace cannot be optimal, because the optimal execution trace of length $3$ does not start with the above $s_1$ and $s_2$.

However, as it turns out, it is possible to obtain a power management policy that is near-optimal for any time horizon:
\begin{defn}[Asymptotic optimality]
\label{def:aopt}
We say a power management policy is \emph{asymptotically optimal} if the policy can be halted after an arbitrary number of iterations, and the average power consumption of the resulting execution trace is asymptotically optimal, i.e., the \emph{error} defined as the difference between the achieved average power consumption and the exact optimum (which can only be calculated with the knowledge of $n$) tends to zero.
\end{defn}
In Section~\ref{sec:opt}, we show that an asymptotically optimal power management policy, in fact a very simple one, does exist under a simplifying assumption that $W$ is a continuous function. 

\subsection{Problem Statement}
\label{subsec:prob}

Finally, we restate our goal formulated as a concrete optimization problem:
\begin{quote}
Given a system specified by its power consumption characteristics $P$, delay-workload relation $W$, initial workload $w_1$, and the real-time constraint $T$, design a power management policy that is \emph{asymptotically optimal}, i.e., a power management policy that gives a near-optimal execution trace for any time horizon.
\end{quote}

\section{Power Optimization: a Special Case}
\label{sec:opt}

In this section, we present a provably asymptotically optimal power management policy, focusing on the special case where the delay-workload relation $W$ is a continuous function. This restriction allows us to omit the tedious details required to maintain the mathematical rigor under the general delay-workload relation, leading to a simpler presentation which still exhibits all the key intuition. A proof with the full generality is deferred to Section~\ref{sec:genproof} for interested readers.

\subsection{Overview}
\label{sec:opt:ov}

Unless the given system is real-time infeasible (this can happen if the system parameters are such that its workload will ``explode'' even if the system is run at the full speed, destined to violate the real-time constraint), the system has ``steady states'' where the workload, execution delay, and processing speed all remain the same across iterations.

Among all the possible steady states, our power management policy finds one that consumes the least amount of power, and stay in this minimum-power steady state.
Intuitively, the convexity of the power consumption characteristics penalizes fluctuation in the speed; the (asymptotic) optimum can therefore be achieved by such a steady state.

In presenting the proposed policy, we first show in Section~\ref{subsec:pre} how to determine the range of execution delays that can lead to a steady state and find a minimum-power steady state. Then, our policy is to quickly enter this minimum-power steady state and remain there; this policy is drawn in Section~\ref{subsec:pol}.
The formal proof of its optimality follows in Section~\ref{subsec:anal}.


\subsection{Preliminaries}
\label{subsec:pre}

\mypar{Bounding execution delays}
As a preparatory step to describe the proposed power management policy, we first bound the range of execution delays. It is relatively easy to see the delay of each iteration is within the range of $[\wbase,T]$;
but here we present a tighter\footnote{In fact, Observation~\ref{obs:lb} is \emph{almost} tight: see Theorem~\ref{thm:reach} for a complete characterization.} lower bound on execution delays, given by Definition~\ref{def:lb}.
\begin{defn}[Minimum delay]
\label{def:lb}
The minimum delay $t_{\min}$ is defined as the longest delay $t\leq w_1$ such that $W(t)$ is greater than or equal to $t$. That is,
$t_{\min} :=\max\{t\mid W(t)\geq t,\ 0<t\leq w_1\}$.
\end{defn}

$t_{\min}$ safely bounds the execution delay of each iteration from below, as observed below.

\begin{obs}
\label{obs:lb}
No execution trace has an iteration whose delay is strictly smaller than $t_{\min}$.
\end{obs}
\begin{proof}

Let $s_1,\ldots,s_n$ be an arbitrary execution trace with execution delays $t_1,\ldots,t_n$. We will show $t_i\geq t_{\min}$ for $i=1,\ldots,n$ by induction. 

\emph{Basis.} 
We have $t_1=\frac{w_1}{s_1}\geq w_1\geq t_{\min}$, where the last inequality follows from the choice of $t_{\min}$. 

\emph{Inductive step.} 
Now suppose $t_{i_0}\geq t_{\min}$ for some $i_0\in\{1,\ldots,n-1\}$. Then $t_{i_0+1}\geq W(t_{i_0})\geq W(t_{\min})\geq t_{\min}$, where the second inequality follows from the monotonicity of $W$ and the third from the choice of $t_{\min}$.
\end{proof}

Observation~\ref{obs:lb} 
is illustrated in Fig.~\ref{fig:tmin} showing how $t_{\min}$ bounds execution delays from below in a typical case. Geometrically speaking, $t_{\min}$ (usually) is the rightmost point of intersection between $w=t$ and $w=W(t)$, restricted to the left of $t=w_1$. For a given initial workload $w_1$, the execution delay never goes below $t_{\min}$ because the workload never goes below $t_{\min}$, even if the system is run at the fastest possible speed.

\begin{figure}[tb]
\centering
\includegraphics[width=6cm, trim =0 0.5cm 0 0cm]{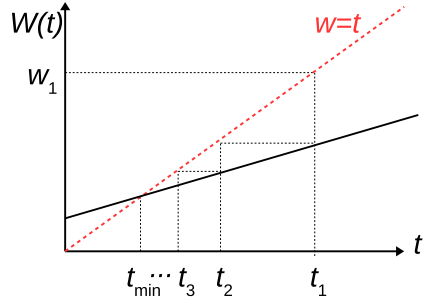}
\caption{An example of the lower bound $t_{\min}$ on execution delays.}\vspace{-1em}
\label{fig:tmin}
\end{figure}

\mypar{Determining the target speed}
From the range of execution delays obtained above, now we define the \emph{target speed} of our policy as follows. 
\begin{defn}[Target speed]
\label{def:khat}
Within the specified range of execution delay, the target speed $\hat s$ is defined as the minimum workload-to-delay ratio in a ``steady state''. That is,
$\hat s:=\min_{t_{\min}\leq t\leq T}\frac{W(t)}{t}$.
\end{defn}
In order to understand the name ``steady state'', suppose that the system enters an iteration with delay $t_i\in \argmin_{t_{\min}\leq t\leq T}\frac{W(t)}{t}$; then, as long as the operating speed is set at $\hat s$, the execution delay will remain the same since $t_{i+1}=\frac{W(t_i)}{\hat s}=\frac{t_i\hat s}{\hat s}$.

\mypar{Sustainability}

One could naturally ask: what if the target speed $\hat s$ is out of the valid range, i.e., $\hat s>1$? In fact, such a system is not ``sustainable'': it fails to respect the real-time constraint $T$ after a bounded number of iterations, as shown by Observation~\ref{obs:sustain}. We will thus assume $\hat s\leq 1$ in what follows. Note that we have $\hat s>0$.
\begin{defn}[Sustainability]
We say a system is \emph{sustainable} if, for all $n$, there exists a real-time feasible execution trace of length $n$.
\end{defn}

\begin{obs}
\label{obs:sustain}
If the target speed is not within the valid range, i.e., $\hat s>1$, the system is not sustainable. In particular, there exists a finite bound $n_0$ such that every execution trace of length $n_0$ or longer violates the real-time constraint.
\end{obs}
\begin{proof}
Choose $n_0:=\lfloor \log_{\hat s}\frac{T}{w_1}\rfloor+2$. We use a similar argument as in the proof of Observation~\ref{obs:lb}. Let $s_1,\ldots,s_n$ be an arbitrary execution trace of length $n\geq n_0$ and $t_1,\ldots,t_n$ be its execution delays. For all $i\geq 2$, we have $t_i= \frac{W(t_{i-1})}{s_i}=\frac{W(t_{i-1})}{t_{i-1}}\frac{t_{i-1}}{s_i} \geq \hat s\cdot\frac{ t_{i-1}}{s_i}\geq \hat s\cdot t_{i-1}$, where the first inequality follows from the definition of $\hat s$ and the second from $s_i\leq 1$; thus, by induction, we have $t_{n_0}\geq\hat s^{n_0-1}\cdot w_1>T$.
\end{proof}

An example of a system that is not sustainable is shown in Fig.~\ref{fig:tmax}. The curve $w=W(t)$ is always above the line $w=t$, except for the ``irrelevant'' portion on the right of the real-time constraint. Therefore, $\hat s$ is greater than 1. In this system, starting with the initial workload of $w_1$, the amount of workload keeps growing even at the full processing speed, eventually violating the real-time constraint $T$. Fig.~\ref{fig:tmax} shows that, even though the system is run at the maximum speed, a real-time constraint violation happens at the fourth iteration, i.e., $t_4>T$.

\begin{figure}[tb]
\centering
\includegraphics[width=5.9cm, trim =0 0.5cm 0 0cm]{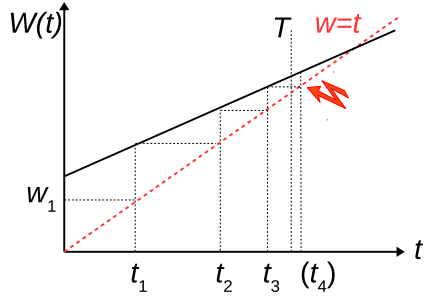}
\caption{An unsustainable system.}\vspace{-1em}
\label{fig:tmax}
\end{figure}

\subsection{Proposed Policy}
\label{subsec:pol}

Our power management policy operates in three simple phases. During the first phase, the system is fixed at the full speed, and the phase lasts until the workload drops below or equal to $\tau:= \max\{t\mid W(t)=\hat s  t,\ t_{\min}\leq t\leq T \}$. Intuitively, $\tau$ stands for the workload at the ideal point where the minimum workload-to-delay ratio $\hat s$ is achieved. Once the workload drops below this level, we can adjust the speed of the system to enter this ideal point. Note that the first phase may be of zero length: our policy immediately enters the second phase if $w_1\leq\tau$. Let $\bar w$ denote the workload at the beginning of the second phase, and then the second phase consists of a single iteration with speed $\bar w/\tau$. Subsequently, the speed is indefinitely fixed at $\hat s$, which is the last phase of our policy.

Fig.~\ref{fig:policy} depicts an example of the execution trace generated by our power management policy. As the initial workload is larger than $\tau$, we start with running at the full speed $s_1=1$, which is the beginning of the first phase. As indicated by the dashed lines, the full processing speed repeatedly reduces the workload of each iteration, eventually reaching below the ideal point workload $\tau$ at the fourth iteration ($w_4<\tau$). Then, the speed is modulated to make the execution delay equal to $\tau$: i.e., we choose $s_4=w_4 /\tau$. This is the single iteration that forms the second phase. In following iterations, the speed is fixed at $\hat s$ in a steady state until the system is halted. In sum, our power management policy results in an execution trace of $1, 1, 1, w_4/\tau, \hat s, \hat s, \hat s, \cdots$ in this example.

\begin{figure}[tb]
\vspace{2em}\centering
\includegraphics[width=7.6cm, trim =3cm 13cm 3cm 2cm]{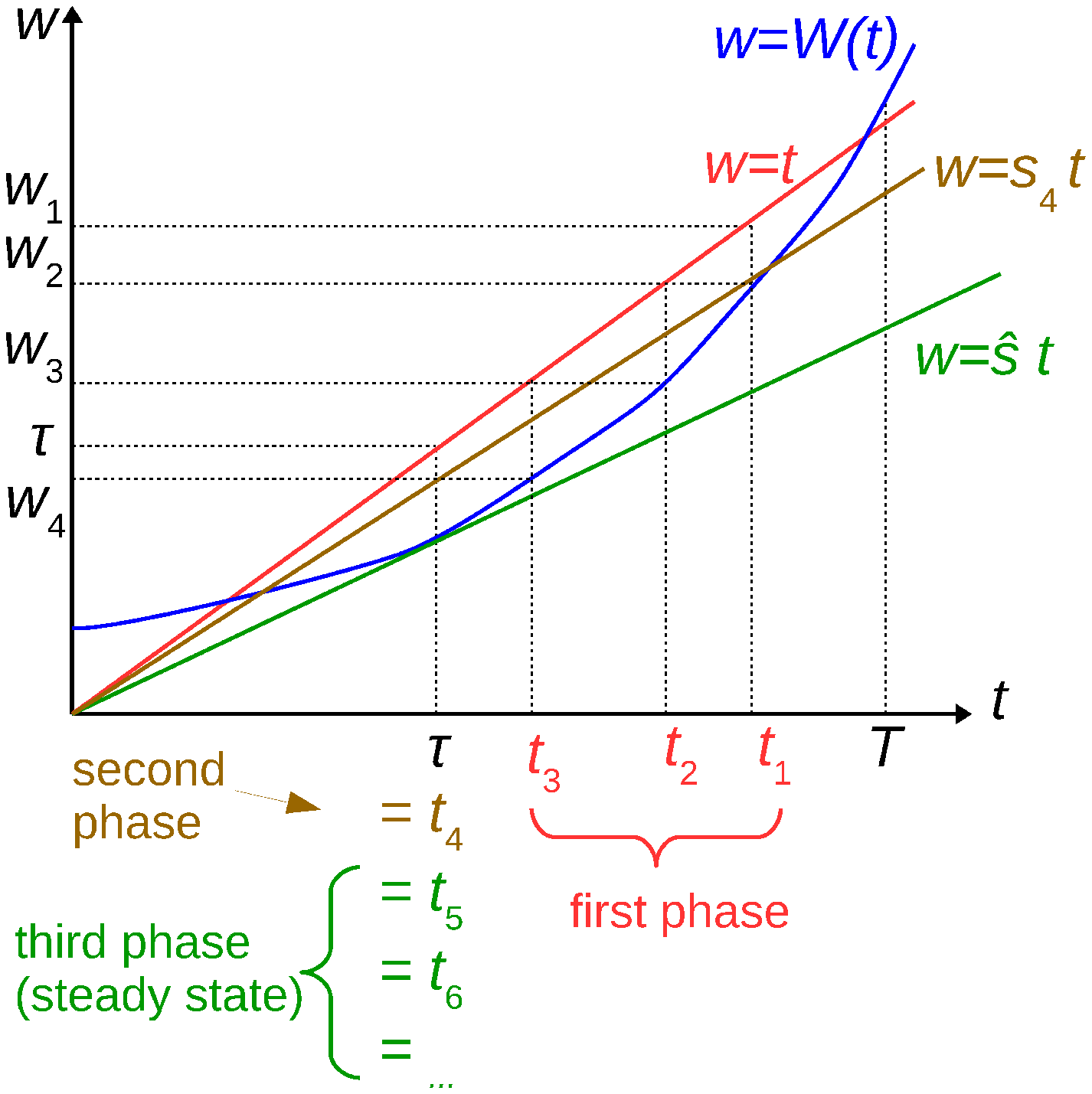}
\caption{An operation example of the proposed power management policy: The first phase runs for three iterations ($t_1,t_2,t_3$) followed by a single iteration ($t_4$) of the second phase. From the fifth iteration on, the system keeps running at the target speed without further modulation of the speed (the third phase). }
\label{fig:policy}
\end{figure}

\subsection{Analysis}
\label{subsec:anal}

In this subsection, we show that our power management policy is asymptotically optimal.
Let $\POL(w_1,n)$ denote the average power consumption of our policy when run for $n$ iterations, and $\OPT(w_1,n)$ denote the infimum\footnote{In fact, the minimum exists if the system is sustainable.} average power consumption of the execution traces of length $n$. In the rest of this section, we show the following main theorem.
\begin{thm}[Asymptotic optimality of the proposed policy]
\label{thm:main}
The difference between the ``exact optimum'' and the average power consumption of our power management policy tends to zero as the time horizon goes to infinity. That is,
$\displaystyle\lim_{n\to\infty} \left[\POL(w_1,n)-\OPT(w_1,n)\right] = 0$.
\end{thm}

Let $s_1,\ldots,s_n$ denote the execution trace of our policy when run for $n$ iterations; $w_1,\ldots,w_n$ and $t_1,\ldots,t_n$ respectively denote the corresponding workloads and execution delays.

\subsubsection{Asymptotic power consumption of our policy}
We first calculate the asymptotic power consumption of our policy.
\begin{lemma}[Asymptotic average power consumption of our policy]
\label{lem:pol}
As the time horizon goes to infinity, the average power consumption of our power management policy converges to that of the target speed $\hat s$. That is, 
$\displaystyle\lim_{n\to\infty}\POL(w_1,n)=\spow(\hat s)$.
\end{lemma}

\begin{claim}[Bounded length of the first phase]
\label{claim:pol1}
For some constant $N$ that does not depend on $n$, the first phase is completed within $N$ iterations. In addition, the first phase does not violate the real-time constraint.
\end{claim}
\begin{proof}
In order to establish the existence of $N$, note that it suffices to show that the first phase is completed within a finite number of iterations: since our algorithm does not assume any knowledge of the time horizon $n$, it is obvious that $N$ does not depend on $n$ as long as the first phase eventually terminates.

There is nothing to prove if $w_1\leq\tau$, since the first phase is then of zero length. Note that this embraces the case where $\hat s=1$: we have $\tau\geq t_{\min}$ by definition, and $t_{\min}=w_1$ when $\hat s=1$ since $W(w_1)\geq \hat s w_1$. Suppose from now that $w_1>\tau$ and $\hat s<1$.

We claim that, for all $t'\in[\tau,w_1]$, $W(t')<t'$. (\textit{Proof.} Suppose there exists $t'\in[\tau,w_1]$ such that $W(t')\geq t'$. Since $W(\tau)=\hat s\tau<\tau$, we have $t'> \tau\geq t_{\min}$ and this contradicts our choice of $t_{\min}$.) Let $m:=\max_{\tau\leq t\leq w_1}\frac{W(t)}{t}$, and we therefore have $m<1$. Intuitively, this $m$ serves as a multiplicative factor that lower bounds the decrease in the workload during the first phase. Thus, it becomes obvious that the first phase eventually terminates. What remains is a rather tedious application of mathematical induction.

Now we show by induction that, if the first phase lasted for at least $\ell$ iterations, $w_{\ell+1}\leq m^\ell \cdot w_1$. The proof is straightforward: firstly, the base case ($\ell=0$) is trivial. If the first phase lasted for at least $\ell=\ell_0$ iterations for $\ell_0\geq 1$, we have $w_{\ell_0}> \tau$ and $s_{\ell_0}=1$  since the $\ell_0$-th iteration is part of the first phase; on the other hand, $t_{\ell_0}=w_{\ell_0}\leq m^{\ell_0-1} \cdot w_1$ holds from the induction hypothesis. This shows $w_{\ell_0+1}=W(t_{\ell_0})\leq W(m^{\ell_0-1}\cdot w_1)\leq m\cdot m^{\ell_0-1}\cdot w_1$, where the first inequality follows from the monotonicity of $W$, and the second from the definition of $m$. Thus, for $N:=\lceil \log_m\frac{\tau}{w_1}\rceil$, the first phase does not last for more than $N$ iterations.

Note that we have $w_1\leq T$, since otherwise it is impossible to meet the real-time constraint in the very first iteration. For each iteration $i$ in the first phase, $t_i\leq m^{i-1}\cdot w_1\leq T$, i.e., the real-time constraint is satisfied.
\end{proof}

\begin{claim}[Real-time feasibility and valid choice of processing speed]
\label{claim:pol2}
Suppose that the time horizon is long enough for the second phase to appear. The processing speed of the second phase is within the valid range, and both the second and third phase of our policy respect the real-time constraint.
\end{claim}
\begin{proof}
Let $i$ be the index of the iteration that constitutes the second phase. We have $s_i=\frac{\bar w}{\tau}\leq 1$ since $\bar w\leq \tau$, and $t_i=\frac{\bar w}{s_i}=\tau\leq T$. The first iteration of the third phase has the workload of $w_{i+1}=W(\tau)=\hat s\tau$, and hence $t_{i+1}=\frac{w_{i+1}}{\hat s}=\tau$; repeating this argument shows that the execution delay of every iteration in the third phase is $\tau\leq T$.
\end{proof}

\begin{proof}[Proof of Lemma~\ref{lem:pol}]
The lemma immediately follows from Claims~\ref{claim:pol1} and \ref{claim:pol2}.

The first two phases have their total execution delay bounded from above by $(N+1)T$, and their average power consumption is no greater than $\spow(1)$. On the other hand, the average power consumption of the third phase is exactly $\spow(\hat s)$, and each of its iteration has execution delay of at least $t_{\min}>0$. Thus, by choosing $n$ to be sufficiently large, $\POL(w_1,n)$ becomes arbitrarily close to $\spow(\hat s)$.
\end{proof}

\subsubsection{Asymptotics of the exact optimum}
Now we analyze the asymptotics of the exact optimum $\OPT$.

\begin{lemma}[Asymptotic optimum]
\label{lem:opt}
As the time horizon goes to infinity, the exact optimum converges to the power consumption of the target speed $\hat s$. That is,
$\displaystyle\lim_{n\to\infty}\OPT(w_1,n)=\spow(\hat s)$.
\end{lemma}
\begin{proof}
Let $\eps>0$ be an arbitrary positive number. Since $\OPT(w_1,n)\leq\POL(w_1,n)$, Lemma~\ref{lem:pol} implies that there exists $N_1\in\mathbb{N}$ such that $\OPT(w_1,n)\leq \POL(w_1,n)<\spow(\hat s)+\eps$ for all $n>N_1$. Hence, it suffices to show that there exists $N_2\in\mathbb{N}$ such that $\OPT(w_1,n)>\spow(\hat s)-\eps$ for all $n>N_2$.

Consider an arbitrary execution trace $s^\mystar_1,\ldots,s^\mystar_n$. Let $P^\mystar$ denote its average power consumption; $w^\mystar_1,\ldots,w^\mystar_n$ and $t^\mystar_1,\ldots,t^\mystar_n$ denote its workloads and execution delays, respectively. We then have\begin{equation}\label{e:lem:opt:1}
P^\mystar
=\frac{\sum_{i=1}^nt^\mystar_i\spow(\frac{w^\mystar_i}{t^\mystar_i})}{\sum_{i=1}^nt^\mystar_i}
\geq\spow(\frac{\sum_{i=1}^nw^\mystar_i}{\sum_{i=1}^nt^\mystar_i})
=\spow(\frac{w_1+\sum_{i=1}^{n-1}W(t^\mystar_i)}{\sum_{i=1}^nt^\mystar_i})
\geq\spow(\frac{\sum_{i=1}^{n-1}W(t^\mystar_i)}{\sum_{i=1}^nt^\mystar_i})
,
\end{equation}where the first inequality follows from the convexity of $\spow$ and the second from monotonicity.

For all $i=1,\ldots,n$, we have $\wbase\leq t^\mystar_i\leq T$; this further implies $W(t^\mystar_i)\leq T$ for all $i=1,\ldots,n-1$. Rewriting \eqref{e:lem:opt:1}, we obtain\begin{eqnarray*}
P^\mystar&\geq&\spow(\frac{\sum_{i=1}^{n-1}W(t^\mystar_i)}{\sum_{i=1}^nt^\mystar_i})\\
&=&\spow(\frac{(\sum_{i=1}^{n-1}W(t^\mystar_i))\cdot[(\sum_{i=1}^{n}t^\mystar_i)-t^\mystar_n]}{(\sum_{i=1}^nt^\mystar_i)\cdot(\sum_{i=1}^{n-1}t^\mystar_i)})\\
&=&\spow(\frac{(\sum_{i=1}^{n-1}W(t^\mystar_i))\cdot\cancel{(\sum_{i=1}^{n}t^\mystar_i)}}{\cancel{(\sum_{i=1}^nt^\mystar_i)}\cdot(\sum_{i=1}^{n-1}t^\mystar_i)}-\frac{t^\mystar_n(\sum_{i=1}^{n-1}W(t^\mystar_i))}{(\sum_{i=1}^nt^\mystar_i)(\sum_{i=1}^{n-1}t^\mystar_i)}
)\\
&=&\spow(\frac{\sum_{i=1}^{n-1}W(t^\mystar_i)}{\sum_{i=1}^{n-1}t^\mystar_i}-\frac{t^\mystar_n\sum_{i=1}^{n-1}W(t^\mystar_i)}{(\sum_{i=1}^{n}t^\mystar_i)\cdot(\sum_{i=1}^{n-1}t^\mystar_i)})\\
&\geq&\spow(\hat s-\frac{t^\mystar_n\sum_{i=1}^{n-1}W(t^\mystar_i)}{(\sum_{i=1}^{n}t^\mystar_i)\cdot(\sum_{i=1}^{n-1}t^\mystar_i)})
,\end{eqnarray*}where
the first equality follows from $\sum_{i=1}^{n-1}t^\mystar_i=(\sum_{i=1}^{n}t^\mystar_i)-t^\mystar_n$, and the last inequality from the monotonicity of $\spow$ and the choice of $\hat s$: note that $\sum_{i=1}^{n-1}W(t^\mystar_i)=\sum_{i=1}^{n-1}s^\mystar_i t^\mystar_i\geq\hat s\sum_{i=1}^{n-1}t^\mystar_i$. Since we have $t^\mystar_n\leq T$, $\sum_{i=1}^{n-1}W(t^\mystar_i)\leq (n-1)T$, and $t^\mystar_i\geq\wbase$ for all $i$, this leads to\begin{eqnarray*}
P^\mystar&\geq&\spow(\hat s-\frac{t^\mystar_n\sum_{i=1}^{n-1}W(t^\mystar_i)}{(\sum_{i=1}^{n}t^\mystar_i)\cdot(\sum_{i=1}^{n-1}t^\mystar_i)})\\
&\geq&\spow(\hat s-\frac{T^2}{n\wbase^2})
,\end{eqnarray*}again from the monotonicity of $\spow$.
Finally, we obtain\[
\OPT(w_1,n)\geq \spow(\hat s-\frac{T^2}{n\wbase^2})
,\]as the above bound holds for any arbitrary execution trace.

On the other hand, since $\spow(s)$ is continuous at $s=\hat s>0$ and nondecreasing, there exists some $\delta>0$ (and $\delta\leq\hat s$) such that $\spow(s)>\spow(\hat s)-\eps$ for all $s\in(\hat s-\delta,1]$. Choosing $N_2:=\frac{T^2}{\delta \wbase^2}$ concludes the proof.

\end{proof}

Theorem~\ref{thm:main} follows from Lemmas~\ref{lem:pol} and \ref{lem:opt}.

\section{Experiments}
\label{sec:case}

In this section, we experimentally evaluate the performance of our power management policy. We perform a case study on a motion/object tracking application to observe how the delay-workload dependence exhibits itself in a CPS. We measure this delay-workload dependence, and apply this to our model along with the power consumption characteristics measured from Exynos5422. We validate the efficiency of our power management policy via comparisons with other approaches.

\subsection{Experimental Setup}
\label{subsec:env}

We choose Odroid-XU3 \cite{odroidxu3} as the target cyber platform, which incorporates Exynos5422 with 2GB main memory running Linux operating system (Ubuntu 14.10). Exynos5422 System-on-a-Chip is a big.LITTLE Octa-core system with 4 big (Cortex-A15) and 4 little cores (Cortex-A9), each of which can be individually and dynamically modulated in operating frequency.

Note that our model does not assume any specific DVFS model and is flexible about the power consumption characteristics used. Thus, rather than resorting to a theoretically derived model of DVFS, we experimentally measure the actual power dissipation of a big core in Exynos5422. In order for this, we pick up a benchmark \emph{blowfish}, which is known to impose a high degree of computational overhead on a CPU \cite{guthaus2001mibench}, and run it on a big core of Exynos5422 repeatedly with different operating frequencies. For precise measurement, we instruct the Linux governor to force the other cores off, and obtain 20 data points as shown in Fig.~\ref{fig:power}. We applied the techniques introduced in Section~\ref{subsec:modelgen:monoconv} in order to cope with the finite number of operation modes. The resulting power characteristics function $P$ is depicted together in Fig.~\ref{fig:power}.

\begin{figure}[tb]
\centering
\includegraphics[width=11.0cm, trim =0 1cm 0 0]{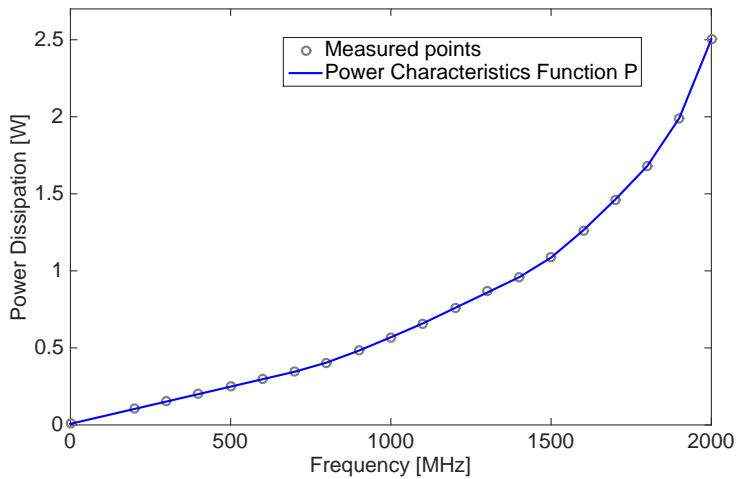}
\caption{Power consumption measurements of a big core in Exynos5422 running blowfish benchmark.}
\label{fig:power}
\end{figure}

\subsection{Delay-Workload Dependence}
\label{subsec:w_exp}

In this subsection, we observe and quantify the delay-workload dependence that emerges in CPS. Motion tracking, or object tracking, is widely used in CPSs~\cite{1656418, 5509287, kanade2004real, coifman1998real, 6155714, sonntag2013towards, you2001fusion} to reflect the changes in the physical world into the internal model that the cyber system maintains, and it is one of the typical sources of delay-workload dependence. We profile the performance of Lucas-Kanade method~\cite{bouguet2001pyramidal}, using the implementation provided in OpenCV Library~\cite{opencv}.

Consider the following scenario in an object tracking CPS, in which the sensor (camera), takes images from the physical world, and the actuator moves the orientation of the camera to enable it to track an object. At the beginning, it is known that the object of interest is located near the center of the scene. As an initial cost, the cyber system takes an image from the physical world through the camera, and selects features to track. This task is done by calling OpenCV APIs \api{cvCvtColor()}, \api{cvGoodFeaturesToTrack()}, and \api{cvFindCornerSubPix()}. The initial workload $w_1$ corresponds to the work done by these function calls.
Later, the system iteratively tracks the target object by repeatedly invoking OpenCV APIs \api{cvCvtColor()} and \api{cvCalcOpticalFlowPyrLK()}. 

Using \emph{a priori} knowledge on the maximum possible speed of the object, we can
calculate the maximum distance that the object could have moved between two iterations. This maximum distance, which is given as an increasing function of the iteration delay, enables us to deduce a bounding box to which the optical flow calculation can safely be limited. In this experiment, we assume that the object speed never exceeds 10~pixels per millisecond; the minimum search area is given as a box of $125\times 125$ pixels. If an iteration takes $t$ milliseconds for instance, the search area for the next iteration is given as a square with the side length of $(10\cdot 2 \cdot t + 125)$ pixels.

We use $(1000\times 1000)$-pixel images to run Lucas-Kanade method and measure the workload incurred. We fix the operating frequency of the big core to the maximum (2GHz), and only one core is activated as the algorithm is single-threaded. We individually measured the workload for search areas of sizes $125\times 125$, $145\times 145$, $\ldots$, $785\times 785$. This provides us with the delay-workload dependence profile shown in Fig.~\ref{fig:workload}.

A piecewise linear, nondecreasing function $W$ is derived from these 34 measurements, as highlighted by the blue curve. We plot a straight line, $w=t$, to provide a visual reference. This line represents the maximum amount of work that can be done by the microprocessor running at the full speed ($s=1$). We also highlight the real-time constraint $T=25\mathrm{ms}$ as a vertical dashed line.

\begin{figure}[tb]
\centering
\includegraphics[width=11.0cm, trim =0 1cm 0 0]{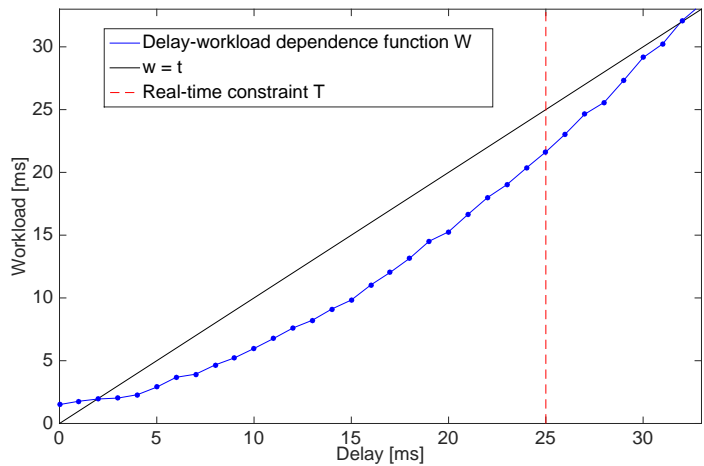}
\caption{Delay-workload dependence observed in Lucas-Kanade method.}
\label{fig:workload}
\end{figure}

\subsection{Power Management Policy}
\label{subsec:p_exp}

In this subsection, we validate the optimality of the proposed power management policy.
We conduct power simulations of Lucas-Kanade algorithm considering the delay-workload dependence provided in Fig.~\ref{fig:workload} with the DVFS modes characterized in Fig.~\ref{fig:power}. The real-time constraint is set to $25\mathrm{ms}$, as shown in Fig.~\ref{fig:workload}.

In order to facilitate comparative investigation, we evaluated four other policies. The first policy is ALAP (see \cite{aydin2004power} for an example of the ALAP approach), where the speed is chosen as the slowest possible while satisfying the real-time constraint. The second policy is ASAP, the other extreme: under this policy, we fix the processing speed to the maximum. The last set of policies, denoted as Heuristic~1 and Heuristic~2, are described in Section~\ref{subsec:heuristic}.

We simulate the five power management policies, i.e., ASAP, ALAP, Heuristic~1, Heuristic~2, and the proposed policy, for six different time horizons: 10, 100, 1000, 10000, 100000, 1000000.
The only exception is Heuristic 2, which violates the real-time constraint at the 67th iteration and therefore is prematurely halted for the last five cases.
We also remark that none of these policies, of course, did not know in advance when they would be terminated.

Table~\ref{tab:pwr} shows the average power consumption achieved by the five power management policies.
Our policy starts outperforming all other approaches when $n\geq100$. In the case of a short time horizon on the other hand ($n=10$), ALAP showed the best average power consumption.

In ASAP, the speed is always set to the maximum, thus the delay converges to $t_{\min}=1.94222\mathrm{ms}$ (see Definition~\ref{def:lb}). In other words, ASAP is the most responsive policy. ALAP, on the contrary, sacrifices the responsiveness in favor of power efficiency; it enters a steady state whose execution delay is exactly equal to the real-time constraint to achieve an asymptotic power consumption of $1.52927\mathrm{W}$. Note that, however, the power efficiency of ALAP is inferior to the proposed policy since it fails to find the better steady state that the proposed policy uses.

The target speed of the proposed policy is identified as $\hat s= 0.56085$. We can observe that, as the time horizon gets longer, the average power consumption gets smaller, converging to $\spow(\hat s)= 0.68087$. As mentioned above, Heuristic~1 was prematurely halted due to real-time constraint violation; Heuristic~2, on the other hand, did not violate the real-time constraint and its average power consumption appears to stay around $\sim 1.15$. See Section~\ref{subsec:heuristic} for further discussion on these heuristics.

\begin{table}[t]
\centering
\caption{Comparison of average power consumptions}
\label{tab:pwr}
\vspace{1ex}\small
\begin{tabular}{l}
\begin{tabular}{r|ccccc}
\hline
 \diagbox{$n$}{{Policy}} &  ASAP & ALAP & Heuristic 1 & Heuristic 2 & \textbf{Proposed} \\ \hline\hline
         10& 2.50325& 1.51550& 2.33322$\phantom{^*}$& 2.36455& \textbf{1.72204}	\\ \hline
        100& 2.50325& 1.52789& 1.12819$^*$          & 1.47065& \textbf{0.82071}	\\ \hline
       1000& 2.50325& 1.52913& 1.12819$^*$          & 1.19146& \textbf{0.69535}	\\ \hline
      10000& 2.50325& 1.52925& 1.12819$^*$          & 1.15731& \textbf{0.68232}	\\ \hline
     100000& 2.50325& 1.52926& 1.12819$^*$          & 1.15465& \textbf{0.68102}	\\ \hline
    1000000& 2.50325& 1.52927& 1.12819$^*$          & 1.15434& \textbf{0.68088}	\\ \hline
\end{tabular}
\vspace{1ex}\\
$^*$Prematurely halted due to the real-time constraint violation.\hfill\mbox{ }
\end{tabular}
\end{table}

\subsection{Heuristics}
\label{subsec:heuristic}

In this subsection, we explore the possibilities of extracting some key ideas that underlie our power management policy, and applying it to devise heuristics that can be used in varied settings. In particular, we will consider a setting in which the power management policy is deprived of its access to the model parameters including the power consumption characteristics $P$ and the delay-workload relation $W$. In fact, the only parameter we assume that the heuristics will be aware of is the real-time constraint $T$. We present a simple heuristic that is inspired by the present power management policy, and works under this limited setting.

Our power management policy, in one line, is to ``find a steady-state point $t$ that minimizes $\frac{W(t)}{t}$''. Even though heuristics under the limited setting does not have access to $W$, it can ``retrospectively'' estimate it: once an iteration, say the $i$-th iteration, completes, we can estimate $W(t_{i-1})=w_i=s_it_i$ by measuring the execution delay $t_i$. Based on this, we can devise the following simple heuristic.

Let $\sigma(t):=\frac{W(t)}{t}$. Recall that our ``objective'' is to minimize $\sigma(t)$. With this in mind, our heuristic estimates $\sigma_i:=\sigma(t_i)$ by $\sigma_i=\frac{W(t_i)}{t_i}=\frac{s_{i+1}t_{i+1}}{t_i}$. Thus, at the end of iteration $i$, the last $\sigma$ we can estimate is $\sigma_{i-1}$.

The heuristic is quite simple: we start with the full speed ($s_i=1$). At any point of the execution, the heuristic has its internal ``intention'' about whether it wants to increase or decrease the processing speed. Initially, this intention is set to \emph{decreasing} since it is the only choice. When a new iteration begins, we reassess the intention. We compare the last two estimates of $\sigma$ to see if we are ``happy'', i.e., $\sigma$ is decreasing, or ``unhappy'', i.e., $\sigma$ is increasing. If we are happy, we do nothing; if unhappy, we flip our ``intention''. A final piece of adjustment is that we perform this reassessment every three iterations, in order to avoid noisy behaviors resulting from the heuristic being too sensitive.

Fig.~\ref{fig:h1} shows the performance of this heuristic, called Heuristic~1, where we increase or decrease the processing speed by $0.01$ at each iteration. As can be seen from the figure, Heuristic~1 rapidly decreases the processing speed to approach the ``true optimum''.
However, it unfortunately overshoots the optimum and tries to recover from the 55th iteration, but eventually fails in the 67th iteration as it violates the real-time constraint.
Taking a closer look, we can see that the policy increases the processing speed until the 58th iteration, at which point it falsely attributes the deteriorating $\sigma$ to its intention and flips it; even though the policy again starts increasing the processing speed at the 64th iteration, the workload has already grown too large by this point and the policy thereby fails to recover.
This failure is largely due to the limitation imposed by the lack of information: without the complete view of the system parameters, the heuristic fails to timely and properly act to recover when it overshoots the optimum.

\begin{figure}[tb]
\centering
\includegraphics[width=12cm, trim =0 1cm 0 1cm]{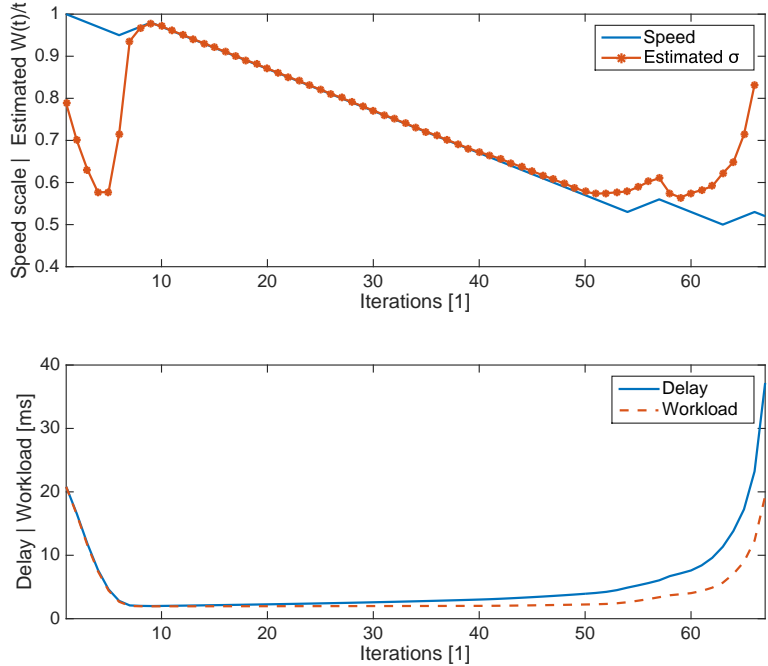}
\caption{Trajectory of speed, estimated $\sigma$, delay, and workload in the execution of Heuristic~1.}
\vspace{-.5em}
\label{fig:h1}
\end{figure}

In order to fix this issue of the belated action, we modify Heuristic~1 by making the increase/decrease in the processing speed asymmetric: when the modified heuristic decreases the processing speed, it decreases it by $0.01$, but when it increases, it does so by $0.1$. This compensates for the belated action by favoring ``escaping out of overload'' over ``making maximum use of low speed''. Fig.~\ref{fig:h2} shows the performance of this modified heuristic, called Heuristic~2, when it is run for 600 iterations until halted from outside. We can see that now the real-time constraint is respected during the entire test run.

\begin{figure}[tb]
\centering
\includegraphics[width=12cm, trim =0 1cm 0 1cm]{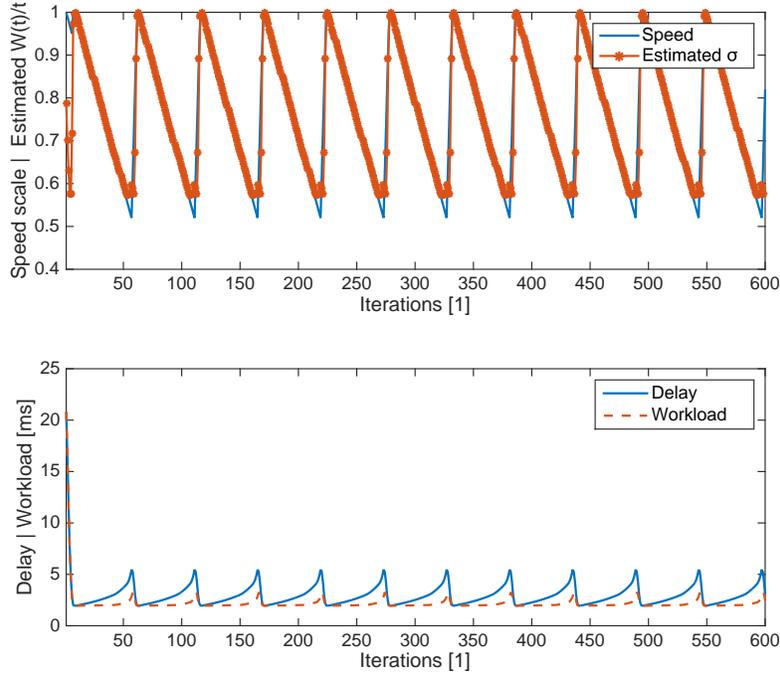}
\caption{Trajectory of speed, estimated $\sigma$, delay, and workload in the execution of Heuristic~2.}
\vspace{-.5em}
\label{fig:h2}
\end{figure}

\section{Power Optimization: the General Case}
\label{sec:genproof}

\subsection{Introduction}
\label{subsec:intro_gen}

In Section~\ref{sec:opt}, we presented a power management policy and its analysis under the simplifying assumption that $W$ is continuous. In this section, we present our results for the general case. This generalization is a key enabler for the application of the proposed policy to the multi-media application domain, where the workload changes in the unit of block in a discrete manner.

Another important use of this generalized policy is as a lightweight substitution of the continuous-case policy of Section~\ref{sec:opt}.
Obtaining an accurate delay-workload relation is a costly operation even though it needs to be performed only once at design time.
In order to save this effort, one can profile the given system only for a few data points to obtain a ``safe upper bound'' of the true delay-workload relation.
Since this upper bound needs to be conservative, a reasonable approach would be to extend these data points into a discontinuous staircase function.
Section~\ref{subsec:genexp} illustrates the operation of our generalized policy in the context of this usage.

The overall organization of this section is quite similar to that of Section~\ref{sec:opt}; in fact, we can draw an almost one-to-one correspondence between the two sections. Recall that the analysis of our policy in Section~\ref{sec:opt} started with bounding the range of execution delays (Observation~\ref{obs:lb}); this section starts with the same, except that it is slightly generalized to cope with discontinuity (Theorem~\ref{thm:reach}). Corollary~\ref{cor:algreach} in the present section shows that the first and second phase still have a bounded length, generalizing Claim~\ref{claim:pol1} of Section~\ref{sec:opt}. The target speed $\hat s$ is again very similarly chosen, but one technicality that exists only in this general case is that the minimum may not exist. If only the infimum exists, instead of choosing a single target speed, we choose a converging sequence of target execution delays. (Speeds are replaced with delays for technical reasons; see Definition~\ref{defn:tds}.) Then, the generalized policy is basically the same as Section~\ref{sec:opt}: we run at the full speed until we reach the target execution delays. Since the target is now defined as an infinite sequence of execution delays, this may not be a steady-state in general, but the analysis shows that they lead to a ``near-steady'' state, achieving asymptotic optimality.

\subsection{Reachability}
We first need to revise our characterization of execution delays that can appear on execution traces: this subsection presents the strengthened counterparts of Definition~\ref{def:lb} and Observation~\ref{obs:lb}.

\begin{defn}[Reachability]
We say an execution delay $t$ is \emph{reachable} if there exists a real-time feasible execution trace $s_1,\ldots,s_n$ with execution delays $t_1,\ldots,t_n$ such that $t_j=t$ for some $j$.
\end{defn}

Now we give a characterization of the reachable delays. Let $R_o:=\{t\mid 0<t<w_1\textrm{ and }\forall t'\in (t,T]\ W(t')>t\}$ and $R_c:=\{t\mid 0<t\leq w_1\textrm{ and }\forall t'\in [t,T]\ W(t')\geq t\}$. Note that both are nonempty. Let $R:=\left(\cap_{t\in R_o}(t,T]\right)\cap\left(\cap_{t\in R_c}[t,T]\right)$. The following theorem shows that $R$ is the desired characterization.

\begin{thm}[Characterization of reachable delays]
\label{thm:reach}
$\tilde t$ is reachable if and only if $\tilde t\in R$.
\end{thm}
\begin{proof}
($\Rightarrow$, every reachable delay is in $R$.) Consider an arbitrary $\tilde t\notin R$ ($\tilde t\leq T$). This implies that either there exists $t\in R_o$ such that $t\geq \tilde t$ or there exists $t\in R_c$ such that $t>\tilde t$ (or both). Let $s_1,\ldots, s_n$ and $t_1,\ldots,t_n$ be the processing speeds and delays of an \emph{arbitrary} real-time feasible execution trace.

\emph{Case 1.} $\exists t\in R_o\ t\geq \tilde t$. We will show by induction that $t_i\in (t,T]$ for all $i$, which implies that $\tilde t$ is not reachable. The base case is easy, since $t_1\in[w_1,T]\subset (t,T]$. Assuming the claim holds for $i=i_0$, we have $t_{i_0+1}\geq w_{i_0+1}=W(t_{i_0})>t$ as desired.

\emph{Case 2.} $\exists t\in R_c\ t>\tilde t$. We use a similar argument: in this case we show $t_i\in [t,T]$ for all $i$. The base case again is easy to see from $t_1\in[w_1,T]\subset [t,T]$. If the claim holds for $i=i_0$, we have $t_{i_0+1}\geq w_{i_0+1}=W(t_{i_0})\geq t$.

($\Leftarrow$, every delay in $R$ is reachable.) Consider an arbitrary $\tilde t\in R$. Let $\bar R$ be the set of reachable execution delays, and we will show that $\tilde t\in \bar R$. Observe that $t\in\bar R$ implies $t'\in\bar R$ for all $t'\in[t,T]$: since there exists an execution trace with $t_j=t$ for some $j$, scaling $s_j$ by a multiplicative factor of $t/t'$ and truncating the trace at the end of iteration $j$ yield a feasible execution trace with $t_j=t'$. Moreover, $W(\bar R)\cap (0,T]\subset \bar R$. (\textit{Proof.} Suppose that, for some $t\in \bar R$, $W(t)\in (0,T]$. Since $t\in\bar R$, there exists an execution trace that reaches $t$. Truncate this execution trace right after the iteration with delay $t$, and add one more iteration with speed $1$. Note that this new iteration has the execution delay of $W(t)$.)
Finally, $\bar R\neq\emptyset$ since $[w_1,T]\subset \bar R$. Suppose towards contradiction that $\tilde t\notin \bar R$.

\emph{Case 1.} $\bar R=[\min \bar R,T]$. We then have $\tilde t<\min\bar R$ since $\tilde t\notin \bar R$. On the other hand, we have $W(\bar R)\subset [\min\bar R,\infty)$ and therefore $\min\bar R\in R_c$ by definition (note that $\min\bar R\leq w_1$). This gives $\tilde t\notin R$.

\emph{Case 2.} $\bar R=(\inf \bar R,T]$. In this case $\tilde t\leq \inf\bar R$. Since $W(\bar R)\subset (\inf\bar R,\infty)$, we have $\inf\bar R\in R_o$ which in turn implies $\tilde t\notin R$.
\end{proof}

So far we have not used the monotonicity of $W$; we will however use it in what follows in order to devise a clean algorithmic way to reach $t\in R$.

Before we do this, we make some useful observations first. Given a real-time feasible execution trace, increasing its speeds preserves feasibility since it does not increase any execution delays:
\begin{obs}[Closedness of feasibility with respect to speed increase]
\label{obs:full}
Let $s_1,\ldots,s_n$ be a real-time feasible execution trace with delays $t_1,\ldots,t_n$. Suppose we change some of $s_i$'s to $1$, obtaining a new execution trace $s'_1,\ldots,s'_n$ with delays $t'_1,\ldots,t'_n$. We have $t'_i\leq t_i$ for all $i$ (and therefore the new trace also is real-time feasible).
\end{obs}
\begin{proof}
Trivial from the monotonicity of $W$.
\end{proof}

Theorem~\ref{thm:reach} along with this observation gives an algorithmic procedure to reach an arbitrary reachable delay $\tilde t\in R$: fix the system at the full speed until the iteration whose workload $w_i$ drops below or equal to $\tilde t$; choose $s_i:=\frac{w_i}{\tilde t}$ so that the delay $t_i$ becomes exactly $\tilde t$.
\begin{cor}[Algorithmic reachability]
\label{cor:algreach}
The above procedure produces a real-time feasible execution trace that reaches $\tilde t$ in the last iteration.
\end{cor}
\begin{proof}
Consider an arbitrary real-time feasible execution trace that witnesses $\tilde t$. Such a trace is guaranteed to exist by Theorem~\ref{thm:reach}. Now, setting all the speeds to 1 yields an execution trace that is real-time feasible and has an iteration whose delay is at most $\tilde t$ (see Observation~\ref{obs:full}). Truncate this execution trace right after first such iteration, and decrease the speed of this last iteration so that its delay becomes exactly $\tilde t$. Note that this decrease does not harm feasibility, and that this exactly corresponds to the execution trace produced by the above procedure.
\end{proof}

\subsection{Target Speed}
Now the definition of the target speed (corresponding to Definition~\ref{def:khat} of the continuous case) is generalized as follows:
\begin{defn}[Target speed]
\label{def:newtarget}
Within the reachable range of execution delays, the target speed $\hat s$ is defined as the infimum workload-to-delay ratio. That is,
$\hat s:=\inf_{t\in R}\frac{W(t)}{t}$.
\end{defn}
Observation~\ref{obs:sustain} extends to this new generalized definition of $\hat s$: the system is not sustainable if $\hat s>1$. Hence, we will assume $\hat s\leq 1$ in what follows. Note that $\hat s>0$ since $\wbase>0$.

Before we present our power management policy generalized for arbitrary nondecreasing workload functions, it may be helpful to review the continuous variant in a slightly different presentation. Recall that, in the continuous variant, the execution delays of the third phase formed a constant sequence $\tau,\tau,\cdots$, and this achieved the desired power consumption $\spow(\hat s)$. The first two phases were simply to initiate this steady state. Our generalized policy works in basically the same way: the only subtlety comes from the fact that there may be no execution delay $\tau$ whose steady state achieves $\hat s$, i.e., $\min_{t\in R}\frac{W(t)}{t}$ may be undefined. Thus, our new generalized policy will use an infinite sequence that converges to the target speed in lieu of $\tau,\tau,\cdots$.

\begin{defn}[Target delay sequence]
\label{defn:tds}
Let $\eps_s,\eps_t>0$ be two positive parameters to be chosen later. We say an infinite sequence $(\tau_i)_{i\in\mathbb{Z}_{\geq 0}}$ is a \emph{target delay sequence} if the following hold:
\begin{enumerate}
\item $\tau_i\in R$ for all $i$;\label{prop:target:feas}
\item $\frac{W(\tau_i)}{\tau_i}<\hat s+\eps_s$ for all $i$, and $\left(\frac{W(\tau_i)}{\tau_i}\right)_{i\in\mathbb{Z}_{\geq 0}}$ converges to $\hat s$;\label{prop:target:kconv}
\item there exists some $\hat\tau\in \cl(R)$ such that $|\tau_i-\hat\tau|<\eps_t$ for all $i$ and $(\tau_i)_{i\in\mathbb{Z}_{\geq 0}}$ converges to $\hat\tau$.\label{prop:target:tconv}
\end{enumerate}
\end{defn}

While, at first glance, this definition might look more complicated than it actually is, identifying a target delay sequence is in fact very simple for most conceivable applications. For example, if $\min_{t\in R}\frac{W(t)}{t}$ exists, we can simply take an infinite \emph{constant} sequence of $\min_{t\in R}\frac{W(t)}{t}$, which is exactly what we did in the continuous case. If $W$ is discontinuous but piecewise continuous, a target delay sequence can be given as either an infinite constant sequence or a sequence converging to one of the discontinuities, where any such sequence will be admissible as long as the first term starts sufficiently close to $\hat\tau$.

\vspace{-.5em}
\subsection{Proposed Policy}
Let $s_1,s_2,\ldots$ be the infinite sequence whose prefix of length $n$ corresponds to the execution trace of our policy when it is run for $n$ iterations. Likewise, let $t_1,t_2,\ldots$ denote the infinite sequence of its execution delays.

When $\hat s=1$, our policy is simply fixing the system at its full speed: $s_i=1$ for all $i$.

When $\hat s<1$, for some target delay sequence $(\tau_i)_{i\in\mathbb{Z}_{\geq 0}}$ where we choose $\eps_s:=\frac{1-\hat s}{3}$ and $\eps_t:=\eps_s\wbase$, our policy first invokes Corollary~\ref{cor:algreach} to reach $\tau_0$. This corresponds to the first two phases of the continuous counterpart. Let $\baseindex$ denote the index of the iteration we reach $\tau_0$: $t_{\baseindex}=\tau_0$. The processing speeds of the following iterations (corresponding to the third phase) are chosen so that the execution delays from then form prefixes of the target delay sequence, i.e., $s_{\baseindex+i}:=\frac{W(\tau_{i-1})}{\tau_i}$ for all $i=1,2,\cdots$.

\vspace{-.5em}
\subsection{Feasibility}
Let us verify the feasibility of our policy. Firstly, consider the case when $\hat s=1$. Our policy produces a real-time feasible execution trace as long as there exists one, as can be seen from Observation~\ref{obs:full}. In the interest of completeness, we also provide the following characterization by which we can determine the system's sustainability.
\begin{lemma}[Sustainability]
\label{lem:1sustain}
A given system is sustainable if and only if\begin{itemize}
\item there exists $t\in[w_1,T]$ such that $W(t)\leq t$, or
\item there exists $t\in(w_1,T]$ such that $W(t')<t$ for all $t'\in[w_1,t)$.
\end{itemize}
\end{lemma}
\noindent
The proof of this lemma uses a similar argument as Theorem~\ref{thm:reach} and is deferred to the end of this section.

Now we will focus on the case where $\hat s<1$. To begin with, the following lemma shows that a target delay sequence is guaranteed to exist. Its proof is given at the end of this section.
\begin{lemma}[Existence of a target delay sequence]
\label{lem:targetexists}
There always exists a target delay sequence.
\end{lemma}

Finally, it remains to verify that the processing speeds are validly chosen.
\begin{lemma}[Validity of the proposed policy]
Processing speeds chosen by our policy are all valid. That is, $s_i\in(0,1]$ for all $i$.
\end{lemma}
\begin{proof}
It suffices to verify the claim for each $i>\baseindex$ from Corollary~\ref{cor:algreach}. For all $i=1,2,\cdots$, we have\[
s_{\baseindex+i}:=\frac{W(\tau_{i-1})}{\tau_i}=\frac{W(\tau_{i-1})}{\tau_{i-1}}\cdot\frac{\tau_{i-1}}{\tau_{i}}<(\hat s+\eps_s)\cdot\frac{\hat\tau+\eps_t}{\hat\tau-\eps_t}
,\]where the last inequality follows from Definition~\ref{defn:tds}. Since $\min\cl(R)\geq\wbase$, we have $\eps_t\leq\eps_s\hat\tau$; thus,\[
s_{\baseindex+i}<(\hat s+\frac{1-\hat s}{3})\cdot\frac{1+\frac{1-\hat s}{3}}{1-\frac{1-\hat s}{3}}=\frac{(1+2\hat s)(4-\hat s)}{3(2+\hat s)}\leq 1
,\]where the last inequality is verified as follows: let $f:[0,1]\to\mathbb{R}$ be a function such that $f(s)=\frac{(1+2 s)(4- s)}{3(2+ s)}$. Since $f'(s)=\frac{2(1-s)(5+s)}{3(2+s)^2}$, $f$ is nondecreasing; on the other hand, $f(1)=1$.
\end{proof}

\subsection{Asymptotic Power Optimality}
We show that the proposed policy is asymptotically power-optimal, this time for general $W$. Let $\POL(w_1,n)$ denote the average power consumption of our policy when run for $n$ iterations, and $\OPT(w_1,n)$ denote the infimum average power consumption of the real-time feasible execution traces of length $n$. In the rest of this section, we will show the following theorem, which is the generalized counterpart of Theorem~\ref{thm:main}.
\begin{thm}[Asymptotic optimality of the proposed policy]
\label{thm:ncmain}
The difference between the infimum average power consumption and the average power consumption of our power management policy tends to zero as the time horizon goes to infinity. That is,
$\displaystyle\lim_{n\to\infty} \left[\POL(w_1,n)-\OPT(w_1,n)\right] = 0$.
\end{thm}

We again begin with calculating $\lim_{n\to\infty}\POL(w_1,n)$.

\begin{lemma}[Asymptotic power consumption of the proposed policy]
\label{lem:ncpol}
As the time horizon goes to infinity, the average power consumption of our power management policy converges to that of the target speed $\hat s$. That is,
$\displaystyle\lim_{n\to\infty}\POL(w_1,n)=\spow(\hat s)$.
\end{lemma}
\begin{proof}
We claim that $\displaystyle\lim_{n\to\infty} \spow(s_i)=\spow(\hat s)$. (\textit{Proof.} Note that $\displaystyle\lim_{n\to\infty} s_i=\lim_{i\to\infty}\frac{W(\tau_{i-1})}{\tau_{i}}=\lim_{i\to\infty}\frac{W(\tau_{i-1})}{\tau_{i-1}}\cdot\frac{\tau_{i-1}}{\tau_{i}}=\hat s\cdot\frac{\hat\tau}{\hat\tau}=\hat s$. Now the claim holds since $\spow$ is continuous.) Recall that $\POL(w_1,n)=\frac{\sum_{i=1}^n t_i\spow(s_i)}{\sum_{i=1}^n t_i}$, where we have $t_i\in R$ for all $i$ with $\inf R>0$ and $\sup R\leq T$, and $\spow(s_i)\leq \spow(1)$. Thus, for any $\eps>0$, there exists $N\in\mathbb{N}$ such that $|\POL(w_1,n)-\spow(\hat s)|<\eps$ for all $n>N$.
\end{proof}

Finally, we can determine $\lim_{n\to\infty}\OPT(w_1,n)$ by following the proof of Lemma~\ref{lem:opt} verbatim. The proof is thereby omitted.
\begin{lemma}[Asymptotic infimum]
\label{lem:ncopt}
As the time horizon goes to infinity, the infimum average power consumption converges to that of the target speed $\hat s$. That is, 
$\displaystyle\lim_{n\to\infty}\OPT(w_1,n)=\spow(\hat s)$.
\end{lemma}

Theorem~\ref{thm:ncmain} follows from Lemmas~\ref{lem:ncpol} and \ref{lem:ncopt}. We conclude this section with the deferred proofs.

\subsection{Deferred Proofs}
\begin{proof}[Proof of Lemma~\ref{lem:1sustain}]
($\Leftarrow$) Let $n\geq 1$ be an arbitrary integer. Consider an execution trace of length $n$ where the speed is fixed at the full speed. Let $t_i$ be the execution delay of the $i$-th iteration.

If there exists $t\in[w_1,T]$ such that $W(t)\leq t$, it is easy to show by induction that $t_i\leq t$ for all $i$. Likewise, if there exists $t\in (w_1,T]$ such that $W(t')<t$ for all $t'\in[w_1,t)$, we can show $t_i<t$ for all $i$.

($\Rightarrow$) 
Suppose that the system is sustainable; Observation~\ref{obs:full} implies that the execution delay of length $n$ in which the speed is fixed at the full speed is real-time feasible. Also note that the execution delay of the $i$-th iteration under an execution trace fixed to the full speed does not depend on the time horizon. Let $t_i$ be this delay, and we obtain an infinite sequence $t_1,t_2,\cdots$.

\emph{Case 1.} $\exists i\ t_{\sup}=t_i$. We have $W(t_i)\leq t_i$, since otherwise $t_{i+1}=W(t_i)>t_i$. Note that $t_{\sup}\in[w_1,T]$.

\emph{Case 2.} $\nexists i\ t_{\sup}=t_i$. Since $t_1=w_1$, we have $t_{\sup}>w_1$. We will prove by contradiction that $W(t')<t_{\sup}$ for all $t'\in[w_1,t_{\sup})$. Suppose there exists $t'\in[w_1,t_{\sup})$ such that $W(t')\geq t_{\sup}$. From the choice of $t_{\sup}$, there exists $j$ such that $t_j\in(t',t_{\sup})$. We then have $t_{j+1}=W(t_j)\geq W(t')\geq t_{\sup}$, reaching contradiction.

\end{proof}

\begin{proof}[Proof of Lemma~\ref{lem:targetexists}]
Since $\hat s:=\inf_{t\in R}\frac{W(t)}{t}$, we can choose a sequence $(t_i)_{i\in\mathbb{Z}_{\geq 0}}$ so that $t_i\in R$ for all $i$ and $\left(\frac{W(t_i)}{t_i}\right)_{i\in\mathbb{Z}_{\geq 0}}$ converges to $\hat s$. For sufficiently large $N$, $(t_i)_{i\geq N}$ yields a sequence that satisfies Properties~\ref{prop:target:feas} and \ref{prop:target:kconv} of Definition~\ref{defn:tds}. From (the one-dimensional case of) Bolzano--Weierstrass theorem (see e.g.~\cite[pp.~54--56]{Apostol}), there exists a subsequence of $(t_i)_{i\geq N}$ which converges, say, to $\hat\tau$. We can then choose a subsequence of this subsequence to achieve Property~\ref{prop:target:tconv} as well. Note that we did not lose Property~\ref{prop:target:feas} or \ref{prop:target:kconv} during this construction.
\end{proof}

\subsection{Experiments}
\label{subsec:genexp}

In this subsection, we experimentally illustrate the operation of our proposed policy under a discontinuous delay-workload relation.

\subsubsection{Setup} 
As was discussed in Section~\ref{subsec:intro_gen}, we use a staircase delay-workload relation obtained from a few data points, which can be an inexpensive substitute of the exact delay-workload relation.
Experimental setup is identical to the experiment in Section~\ref{sec:case}; the only difference is that we use only 11 profiling data points in this experiment. These 11 data points are 3ms apart, i.e., we use the profiling results for execution delays of 0, 3ms, 6ms, $\ldots$, and 30ms.

The conservative staircase bound $W$ is formally defined as follows: given $k$ data points $(\bar t^1,w^1),\ldots,\linebreak (\bar t^k,w^k)$, let $W(t):=\min_{i:t\leq \bar t^i} \bar w^i$. Our staircase delay-workload relation is depicted in Fig.~\ref{fig:genworkload}; this function is presented as the delay-workload relation to each simulated policy (the two heuristics, of course, do not get any information on the delay-workload relation).
These policies therefore work with a conservative bound of the workload, but our simulation uses the true delay-workload relation shown in Fig.~\ref{fig:workload} to calculate the actual execution delays.

\subsubsection{Results}
Table~\ref{tab:genpwr} shows the average power consumption of each policy, where each policy is terminated after 10, 100, 1000, 10000, 100000, and 1000000 iterations. The two heuristic does not assume any knowledge on $W$ and therefore behaves identical to the first experiment (recall that our simulation relied on the true delay-workload relation). Since ASAP simply fixes the processing speed to $s=1$, its result is identical to the first experiment as well.
ALAP, on the other hand, behaves differently because what the policy thinks is the slowest possible processing speed subject to the real-time constraint actually is faster than the true value. We can observe that its average power consumption fluctuates, approximately around 1.43. Our policy, finally, converges to the target speed of $\hat s=0.58107$ (the target delay sequence can be chosen as a constant sequence in this case, see $w=\hat s\cdot t$ in Fig.~\ref{fig:genworkload}.) and its average power consumption tends to $\spow(\hat s)=0.72119$.

\begin{table}[t]
\centering
\caption{Comparison of average power consumptions \mbox{(with a staircase delay-workload relation)}}
\label{tab:genpwr}
\vspace{1ex}\small
\begin{tabular}{l}
\begin{tabular}{r|ccccc}
\hline
 \diagbox{$n$}{Policy} &  ASAP & ALAP & Heuristic 1 & Heuristic 2 & \textbf{Proposed} \\ \hline\hline
         10& 2.50325& 1.42631& 2.33322$\phantom{^*}$& 2.36455& \textbf{1.58071}	\\ \hline
        100& 2.50325& 1.44927& 1.12819$^*$          & 1.47065& \textbf{0.82424}	\\ \hline
       1000& 2.50325& 1.43451& 1.12819$^*$          & 1.19146& \textbf{0.73150}	\\ \hline
      10000& 2.50325& 1.43061& 1.12819$^*$          & 1.15731& \textbf{0.72222}	\\ \hline
     100000& 2.50325& 1.43215& 1.12819$^*$          & 1.15465& \textbf{0.72129}	\\ \hline
    1000000& 2.50325& 1.43233& 1.12819$^*$          & 1.15434& \textbf{0.72120}	\\ \hline
\end{tabular}
\vspace{1ex}\\
$^*$Prematurely halted due to the real-time constraint violation.\hfill\mbox{ }
\end{tabular}
\end{table}

\begin{figure}[t]
\centering
\includegraphics[width=12.0cm, trim =0 1cm 0 1cm]{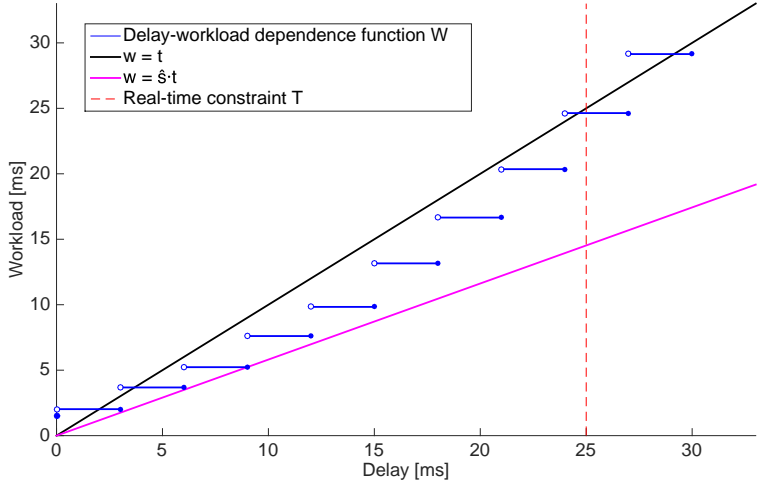}
\caption{Staircase delay-workload dependence in Lucas-Kanade method.}
\label{fig:genworkload}
\end{figure}

\section{Generalizing to Weaker Sets of Assumptions}
\label{sec:modelgen}
Our model proves quite versatile as the set of assumptions made by our model is not minimal: depending on the application at hand, we can drop and/or modify some of these assumptions without loss of generality in order to arrive at a more general model. We discuss these generalizations in this section. It is worth mentioning that they greatly improve the applicability of the proposed policy to a wide variety of underlying hardware platforms.

Sections~\ref{subsec:modelgen:monoconv} and~\ref{subsec:modelgen:discrete} show how the assumptions on the power consumption characteristics $P$ can be relaxed. In particular, Section~\ref{subsec:modelgen:monoconv} explains how to drop the monotonicity and convexity assumptions, whereas Section~\ref{subsec:modelgen:discrete} discuss how to generalize our model to cope with discrete speed modes. Finally, Section~\ref{subsec:modelgen:wmono} explains how the monotonicity assumption on the delay-workload dependence can also be lifted.

\subsection{Dropping the Monotonicity \& Convexity Assumptions on $P$}
\label{subsec:modelgen:monoconv}
Recall that our model assumes that $P$ is nondecreasing, convex, and continuous. In what follows, we show that our model can be generalized to include any continuous $P$.

\subsubsection{Dropping monotonicity}

First we show that the monotonicity assumption can be dropped. Dropping the monotonicity would imply that reducing the speed could cause even bigger power consumption. Intuitively, it is quite clear what we would do in this case: if there are two speed modes $s1$ and $s2$ such that $s1$ is both more power-consuming ($P(s1)>P(s2)$) and slower ($s1<s2$), we will never use $s_1$ and simply replace with $s_2$. The resulting ``new power consumption characteristics'' is denoted by $\bar\spow$ below. We will conclude with a formal argument.

When $P$ is convex and continuous, we show that we can assume without loss of generality that $P$ is nondecreasing. Let $\spow:[0,1]\to\mathbb{R}_+$ be an arbitrary convex and continuous function and $\bar \spow:[0,1]\to\mathbb{R}_+$ be a function defined by $\bar\spow(s):=\min_{s\leq s'\leq 1} P(s')$. It is easy to see that $\bar P$ is continuous and nondecreasing. Moreover, $\bar P$ is convex\footnote{
Consider arbitrary $s_1,s_2,s,\lambda\in [0,1]$ such that $s=\lambda s_1+(1-\lambda) s_2$. For some $s'_1\geq s_1$ and $s'_2\geq s_2$, we have $\bar \spow(s_1)=\spow(s'_1)$ and $\bar \spow(s_2)=\spow(s'_2)$. Observe that $\bar\spow(s)\leq\spow(\lambda s'_1+(1-\lambda) s'_2)\leq\lambda\spow(s'_1)+(1-\lambda)\spow(s'_2)=\lambda \bar\spow(s_1)+(1-\lambda)\bar\spow(s_2)$.

}.

Now we can run the power optimization policy using $\bar\spow$ in lieu of $\spow$, and if the policy says that iteration $i$ is to be run at speed $s_i$ where $\spow(s_i)>\bar\spow(s_i)=\spow(s'_i)$ for some $s'_i>s_i$, we run the iteration at speed $s'_i$ instead. Note that this allows a strictly larger amount of work to be done during the iteration, whereas the power consumption is kept at $\bar\spow(s_i)$. This shows that any power management policy can be used in conjunction with a power consumption characteristics $P$ that is convex and continuous but not necessarily nondecreasing.

\subsubsection{Dropping convexity}

It remains to show that the convexity assumption can further be dropped.
In order to see the intuition first, suppose that three speed modes $s1<s2<s3$ exhibit non-convex power consumption characteristics. Instead of using mode $s2$, we may interleave $s1$ and $s3$ properly, resulting in the same delay as using $s2$. This allows us to regain the convexity assumption. Again, a formal argument follows below.

Let $\spow:[0,1]\to\mathbb{R}_+$ be an arbitrary continuous function. We define $\bar \spow:[0,1]\to\mathbb{R}_+$ as \[
\bar\spow(s):=\min_{s_A,s_B,\lambda\in[0,1], \lambda s_A+(1-\lambda)s_B=s}\lambda\spow( s_A)+(1-\lambda)\spow(s_B)
;\]and we can easily verify that $\bar\spow$ is convex and continuous.

Similarly to above, we run the power optimization policy using $\bar\spow$ in lieu of $\spow$. Suppose that the policy chose speed $s_i$ at iteration $i$, where $P(s_i)>\bar P(s_i)=\lambda\spow( s_A)+(1-\lambda)\spow(s_B)$ for some $s_A,s_B,\lambda\in[0,1]$ such that $\lambda s_A+(1-\lambda)s_B=s$. Running this iteration at speed $s_A$ for $\frac{\lambda w_i}{s_i}$ units of time and at $s_B$ for $\frac{(1-\lambda) w_i}{s_i}$ ensures that the same amount work can be done while the power consumption is kept at $\bar\spow(s_i)$.

\subsection{Discrete Speed Modes}
\label{subsec:modelgen:discrete}
The arguments we used above in order to drop the convexity assumption can in fact be extended to let our model handle the cases where there are a finite number of speed modes. This, for example, covers the case where a microprocessor supports only a predetermined set of frequency-voltage configurations. The underlying intuition is basically the same: again, we emulate an imaginary speed mode by interleaving two existing speed modes.

For a finite set $S\subset[0,1]$ with $0,1\in S$, let $s_1,\ldots,s_{|S|}$ denote the elements in $S$. When the power consumption characteristics is given by $P:S\to\mathbb{R}_+$, defining $\bar\spow:[0,1]\to\mathbb{R}_+$ as\[
\bar\spow(s):=\min_{\sum_{i=1}^{|S|}\lambda_i s_i = s,\sum_{i=1}^{|S|}\lambda_i=1,\lambda_i\in[0,1]\ \mathrm{for}\ \mathrm{all}\ i}\sum_{i=1}^{|S|}\lambda_i\spow(s_i)
\]
yields a convex and continuous function $\bar\spow$.

Similarly to the previous argument, we can run the power management policy in conjunction with $\bar\spow$ in lieu of $\spow$ without loss of generality: running at speed $s$ where $\bar\spow(s)<\spow(s)$ will now be interpreted as using a convex combination of the finitely many (as opposed to two) modes.

\subsection{Replacing the Monotonicity Assumption on $W$}
\label{subsec:modelgen:wmono}
Instead of assuming that $W$ is nondecreasing, we can assume that $W$ is continuous. A non-monotone delay-workload relation means that reducing a delay may lead to a larger workload. In this case, we would naturally just spend more time in this iteration to avoid this anomaly. This is modeled by the new ``imaginary'' delay-workload relation $\bar W$ in what follows.

Define $\bar W:(0,T]\to\mathbb{R}_+$ as $\bar W(t):=\min_{t\leq t'\leq T}W(t')$, and we can easily verify that $\bar W$ is nondecreasing.
Now we can run the power management policy using $\bar W$ instead of $W$; when $W(\frac{w_i}{s_i})>\bar W(\frac{w_i}{s_i})=W(t')$ for some $t'>\frac{w_i}{s_i}$, we introduce $t'-\frac{w_i}{s_i}$ additional units of intentional delay doing nothing. This will increase the delay of the $i$-th iteration to $t'$, assuring that $w_{i+1}$ becomes $\bar W(\frac{w_i}{s_i})$.

\section{Conclusion and Future Work}
We identified a new challenge in the design of CPSs that was previously unheard of in the design of classical real-time systems. The interaction between the cyber and physical components of CPSs induces delay-workload dependence, creating the unique challenge of power optimization with delay-workload dependence awareness. We presented the first formal and comprehensive model, enabling rigorous investigation of this problem. We  proposed a very simple power management policy, and proved this policy is asymptotically optimal. We also experimentally validated the efficiency of our policy.

Our model requires the delay-workload dependence to be determined at design-time using profiling or static analysis. While a complete characterization of workload is \emph{necessary} in the design of a real-time system due to the stringent nature of the real-time constraint, it is rather unfortunate that both profiling and static analysis are operations that can be expensive. In a soft real-time system, on the other hand, the real-time constraint is allowed to be violated ``every once in a while'' and therefore a complete characterization is not an absolute necessity. It would be an interesting future direction to distill ideas from our result to devise a power management policy that operates under an incomplete workload characterization, where the characterization can be obtained at run-time for example as was done by the heuristics we briefly considered in Section~\ref{subsec:heuristic}.
Another interesting future direction would be in incorporating randomness into our framework. In addition to the possible use of randomized computations, CPSs in particular has multiple other sources of randomness, including the physical world. In order to exploit the full potential of power optimization especially under soft real-time settings, it would be useful to allow the parameters of our model to be stochastically specified or adaptively changed.

\section*{Acknowledgment}
The authors would like to thank the anonymous reviewers of versions of this paper~\cite{yang2015modeling, Journalversion} for their helpful comments, and Prof.\@ Seokhee Jeon, Prof.\@ Yong Seok Heo, and Prof.\@ Young-Dae Hong for helpful discussions. 

\bibliographystyle{abbrv}
\bibliography{lit}



\end{document}